\documentclass{article}
\usepackage{arxiv}

\usepackage{hyperref}
\usepackage{microtype}

\usepackage{amssymb}
\usepackage{amsmath}
\usepackage{amsthm}
\usepackage{mathdots}
\usepackage{mathtools}

\usepackage{tensor}

\usepackage{urwchancal}
\DeclareFontFamily{OT1}{pzc}{}
\DeclareFontShape{OT1}{pzc}{m}{it}{<-> s * [1.10] pzcmi7t}{}
\DeclareMathAlphabet{\mathpzc}{OT1}{pzc}{m}{it}

\usepackage{graphicx}
\usepackage{ifpdf}
\ifpdf
\usepackage{epstopdf}
\PrependGraphicsExtensions{.svg}
\fi

\usepackage{xypic}

\newtheorem{theorem}{Theorem}[section]
\newtheorem{lemma}[theorem]{Lemma}

\newtheorem{remark}[theorem]{Remark}

\usepackage{url}

\providecommand{\N}{\mathbb{N}}
\providecommand{\R}{\mathbb{R}}

\providecommand{\SO}{\mathbf{SO}}

\providecommand{\SE}{\mathbf{SE}}

\providecommand{\grpG}{\mathbf{G}}

\providecommand{\gothg}{\mathfrak{g}}

\providecommand{\so}{\mathfrak{so}}

\providecommand{\calM}{\mathcal{M}}

\providecommand{\vecV}{\mathbb{V}}

\providecommand{\Id}{I} 

\providecommand{\td}{\mathrm{d}}
\providecommand{\tD}{\mathrm{D}}

\usepackage{accents}
\makeatletter
\providecommand{\scirc}{%
    \hbox{\fontfamily{\rmdefault}\fontsize{0.4\dimexpr(\f@size pt)}{0}\selectfont{\raisebox{-0.52ex}[0ex][-0.52ex]{$\circ$}}}}

\makeatother

\mathchardef\mhyphen="2D

\providecommand{\etal}{\textit{et al.}~}

\newcommand{\publicationdetails}
{
	\copyrightNoticeIEEE{2019}
	Ng, Y., Van Goor, P., Mahony, R., and Hamel, T. (2019), "Attitude Observation for Second Order Attitude Kinematics", \textit{2019 IEEE 58th Conference on Decision and Control (CDC)}, Nice, France, 2019, pp. 2536–2542,
	\DOILink{https://ieeexplore.ieee.org/abstract/document/9029785}{10.1109/CDC40024.2019.9029785}
}
\newcommand{\publicationversion}
{Author accepted version}

\begin{document}

\title{Attitude Observation for Second Order Attitude Kinematics}
\headertitle{Attitude Observation for Second Order Attitude Kinematics}

\author{
\href{https://orcid.org/0000-0002-7764-298X}{\includegraphics[scale=0.06]{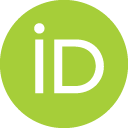}\hspace{1mm}
Yonhon Ng}
\\
Systems Theory and Robotics Group \\
Australian National University \\
ACT, 2601, Australia \\
\texttt{yonhon.ng@anu.edu.au} \\
\And	
\href{https://orcid.org/0000-0003-4391-7014}{\includegraphics[scale=0.06]{orcid.png}\hspace{1mm}
Pieter van Goor}
\\
Systems Theory and Robotics Group \\
Australian National University \\
ACT, 2601, Australia \\
\texttt{pieter.vangoor@anu.edu.au} \\
\And	
\href{https://orcid.org/0000-0002-7803-2868}{\includegraphics[scale=0.06]{orcid.png}\hspace{1mm}
Robert Mahony}
\\
Systems Theory and Robotics Group \\
Australian National University \\
ACT, 2601, Australia \\
\texttt{robert.mahony@anu.edu.au} \\
\And	
\href{https://orcid.org/0000-0002-7779-1264}{\includegraphics[scale=0.06]{orcid.png}\hspace{1mm}
Tarek Hamel}
\\
I3S (University C\^ote d'Azur, CNRS, Sophia Antipolis)\\
and Insitut Universitaire de France\\
\texttt{thamel@i3s.unice.fr} \\
}

\maketitle

\begin{abstract}
	This paper addresses the problem of estimating the attitude and angular velocity of a rigid object by exploiting its second order kinematic model.
	The approach is particularly useful in cases where angular velocity measurements are not available and the attitude and angular velocity of an object need to be estimated from accelerometers and magnetometers.
	We propose a novel sensor modality that uses multiple accelerometers to measure the angular acceleration of an object as well as using magnetometers to measure partial attitude.
	We extend the approach of equivariant observer design to second order attitude kinematics by demonstrating that the special Euclidean group acts as a symmetry group on the system considered.
	The observer design is based on the lifted kinematics and we prove almost global asymptotic stability and local uniform exponential stability of the estimation error.
	The performance of the observer is demonstrated in simulation.
\end{abstract}


\section{INTRODUCTION}
This paper explores the estimation of both the attitude and angular velocity of a rigid object using only accelerometer and magnetometer sensors attached rigidly at known locations on the body.
In particular, the approach taken does not require gyroscope measurements of angular velocity.
The solution to this problem has numerous applications where deployment of gyroscope MEMS devices is not possible due to size constraints, speed of rotation (saturating a MEMS gyroscope), impact characteristics (destroying/destabilising a MEMS gyroscope), vibration, \textit{etc}.
Example applications include complex acrobatic maneuvers of unmanned aerial vehicles~\cite{Lee10}\cite{Richter16}, detumbling maneuvers for satellite deployment~\cite{Ahmed14}\cite{Shou10}, and tracking of fast-moving sporting equipment~\cite{Zhang12}. \cite{Lee10}
A gyroscope is also more prone to malfunction than accelerometer and magnetometer devices, a famous example being the Hubble space telescope that required expensive replacement of the faulty gyroscopes in two separate occasions~\cite{Berkane18}.

The subject of attitude estimation has a rich history of research~\cite{Markley05}.
Different kinematic models and attitude representations have been considered including; Euler angles~\cite{Foxlin96}\cite{Setoodeh04}\cite{Han11}, quaternions~\cite{Lefferts82,Martin10}, and rotation matrices~\cite{Mahony08}\cite{Mahony09}.
Euler angles provide a minimal state representation, however, they have a  singularity problem when the pitch angle passes through $\pm \pi/2$~\cite{Stuelpnagel64}.
Quaternion representations have an ambiguity in the sign, where the same rotation can be represented as $q$ or $-q$.
With the power of modern compute hardware, even for embedded applications, the computational advantages of using the Euler angle and quaternion representations are marginal and the direct matrix representation provides a powerful and simple formulation that we exploit in the present paper.
Existing work can also be classified as either filtering-based~\cite{Lefferts82,Setoodeh04}, or nonlinear observer-based~\cite{Mahony08,Martin10,Hua14} methods.
Filtering methods are based on stochastic principles and provide estimates of uncertainty, however, they are complex (involving propagation of covariance matrix estimates) and typically do not provide the same convergence guarantees and simplicity that are provided by observer-based methods.

The majority of existing attitude estimation algorithms rely on the use of gyroscopes to measure the angular velocity of the rotating body.
Lizarralde and Wen~\cite{Lizarralde96} proposed the use of a nonlinear filter of the quaternion to remove the need for angular velocity feedback.
Berkanel \etal~\cite{Berkane18} presented a hybrid attitude and angular velocity observer that does not rely on gyroscope measurement.
These methods however rely on the knowledge of the inertia matrix and input torque, which are then used to compute the attitude and angular velocity.
Recently, Magnis and Petit~\cite{Magnis17} proposed the use of multiple vector measurements to recover the angular velocity.
They extended the state space to include the estimation of the vector measurement along with the angular velocity.

In this paper, we approach the problem from the perspective of nonlinear observer theory and seek a simple and robust innovation that guarantees almost global stability of the attitude estimate (note that global stability for a continuous estimation algorithm is prevented by topological constraints of the state space \cite{2000_Bhat_SCL}). 
We propose the use of multiple strategically placed accelerometers and magnetometers to obtain a robust estimation of a rigid object's (\emph{e.g.} unmanned aerial vehicle (UAV)) attitude and angular velocity, without relying on the gyroscope. 
A key innovation in the paper is to model the second order kinematics of the attitude system (rather than its dynamics) and to show that these kinematics can be treated using the machinery of equivariant observer design. 
The paper makes the following contributions: 
\begin{itemize}
    \item We propose a novel sensor modality that indirectly measures angular acceleration and attitude using multiple accelerometers and magnetometer.  The sensitivity to the angular velocity and acceleration can be adapted for specific application by varying the separation between the accelerometers.
    \item We derive a second order equivariant attitude observer. To our knowledge this is the first deliberate application of the theory of equivariant observer design to a second order kinematic system.
    \item We propose a simple observer for robust attitude and angular velocity estimation that does not require angular velocity measurements.
    \item We prove almost global asymptotic convergence and local uniform exponential convergence of the estimation error for the proposed observer.
\end{itemize}

The paper is organised as follows.
Section I presents the introduction and literature review of related work, Section II discusses the problem formulation where we present the state and velocity space, system kinematics and sensor measurements used.
Section III discusses the symmetry group of the attitude system, group actions and the associated system lift.
Section IV presents our observer design and stability analysis.
Section V presents the simulation results, followed by conclusions in Section VI.

\section{PROBLEM FORMULATION}
Let $\calM$ and $\N_i$ for $i = 1, \dots, p$ be finite-dimensional smooth real manifolds that are termed, respectively, the state and output spaces. Let $\mathbb{V}$ denotes the finite-dimensional real vector space that is termed the input (velocity) space. A kinematic system is defined by the state equations
\begin{subequations}
\label{eq:kinematics}
\begin{align}
    \dot{x} &= f(x, v), \\
    y_i &= h_i(x),
\end{align}
\end{subequations}
for smooth dynamic function $f: \calM \times \mathbb{V} \to T\calM$, with $f(x,\cdot): \mathbb{V} \to T_x \calM$ a linear map, and smooth output maps $h_i : \calM \to \N_i$.

We are interested in the second order attitude kinematics that would normally be expressed on $T \SO(3)$. 
By using the standard left trivialisation of the velocity space, a tangent vector $R \Omega^\times \in T_R \SO(3)$ is identified with $\Omega^\times \in \so(3)$. 
In this manner, we consider the state of our system as an element of the product manifold
\begin{align}
    \calM = \SO(3) \times \so(3). \notag
\end{align}
The associated input (velocity) space is the tangent $T_R \SO(3)$ to $\SO(3)$ at $R \in \SO(3)$ and that tangent space $T \so(3)$ to $\so(3)$ at a point $\Omega^\times \in \so(3)$.  
Since $\so(3)$ is a linear space, then $T \so(3) \equiv \so(3)$. 
The tangent space $T_R \SO(3)$ can be left trivialised to $\so(3)$ analogously to the approach taken above.  
Thus, the input space that we consider is based on this construction for a parametrization of the tangent space to $\calM$.  
That is a space $\vecV = \so(3) \times T \so(3)$ where we preserve the $T \so(3)$ notation to make clear the part of the velocity space that is modelling the second order part of the kinematics.
An element $v = (\pi^{\times} , \theta^{\times}) \in \vecV \equiv \so(3) \times T\so(3)$ in the input space can be thought of as two independent elements $\pi^\times \in \so(3)$ and $\theta^\times \in T \so(3) \equiv \so(3)$. 

For an element $x = (R,\Omega^{\times}) \in \calM$, and input $v = (\pi^{\times} , \theta^{\times}) \in \vecV \equiv \so(3) \times T\so(3)$, the kinematics are given by
\begin{subequations}
\label{eq:full_kinematics}
\begin{align}
    \dot{R} &= R(\Omega^{\times} + \pi^{\times}), \\
    \dot{\Omega}^{\times} &= \theta^{\times}.
\end{align}
\end{subequations}
Note that the natural behaviour of the system is recovered by measuring the angular acceleration input $\theta^\times$ and by setting $\pi^\times \equiv \mathbf{0}$.
That is, the input angular velocity $\pi^\times$ that acts on the first order state kinematics is zero in the natural system kinematics. 
A key property of the kinematics \eqref{eq:full_kinematics} is the presence of the drift term 
\[
f((R,\Omega), \mathbf{0}) = (R \Omega^\times, \mathbf{0}) 
\]
for $v \equiv \mathbf{0} \in \vecV$ . 
That is, the system function is affine in the input $v \in \so(3) \times T\so(3)$ not linear as has been considered in most prior works~\cite{Bonnabel08, Mahony08, Mahony13} in equivariant observer design.

\begin{remark}
The majority of the work on equivariant observer design has been done for first order kinematics systems 
\[
\dot{X} = X V_1 + V_2 X 
\]
where $X \in \grpG$ is in some Lie-group and  $V_1, V_2 \in \gothg$ are some measured velocities.
In such cases there is no drift term, and the system function is linear with respect to the measured velocity.
The more general case of second order kinematics is implicit in some of the early work on velocity-aided attitude \cite{Bonnable09} since the second order velocity kinematics are present in the model. 
However, in these works the attitude kinematics are not modelled, only the second order attitude (velocity) kinematics and the first order attitude (state) kinematics. 
As such the drift term in \eqref{eq:full_kinematics} is not present. 
The authors are not aware of any other works that use equivariant observer techniques for second order kinematic systems. \hfill$\Box$
\end{remark}

The first component of the state $x = (R,\Omega^{\times}) \in \SO(3) \times \so(3)$ is the attitude of the object with respect to an inertial frame.
The columns of $R$ can also be seen as the coordinate basis of the body-fixed frame in the inertial frame. 
The second component of the state $x = (R,\Omega^{\times}) \in \SO(3) \times \so(3)$ is the angular velocity $\Omega$ expressed in the body-fixed frame $\{B\}$.
The input velocity $v$ of the system kinematics
\begin{equation} \label{eq:system}
    \dot{x} = f(x,v) = (R(\Omega^{\times} + \pi^{\times}), \theta^{\times})
\end{equation}
is the input angular velocity and the angular acceleration of the rotating UAV with respect to the inertial frame, expressed in the body fixed frame $\{B\}$. 

\subsection{Measurements of Angular Acceleration}
The measurement system proposed in this paper involves at least four non-coplanar accelerometers and one magnetometer placed on a rigid body at known displacement from the centre of mass.
We denote the body-fixed frame $\{B\}$ and assume that the accelerometers are aligned (in orientation) with $\{B\}$.
In practice, calibration may be required to find a transformational mapping from the orientation of each accelerometer to the body-fixed frame.
Denote the measured acceleration at each accelerometer $a_i$ for $i=1,...,n$. Then we have
\begin{align}
    \label{eq:accelerometer-original}
    a_i &= a_0 + \dot{\Omega} \times r_i + \Omega \times ( \Omega \times r_i)  \\
    &= a_0 + \dot{\Omega}^{\times} r_i + \Omega^{\times} \Omega^{\times} r_i \nonumber
\end{align}
where $r_i$ is the position of accelerometer $i$ with respect to $\{ B \}$, $a_0$ is the linear acceleration of $\{ B \}$, $\dot{\Omega}$ is the angular acceleration of $\{ B \}$, and $\Omega$ is the angular velocity of $\{ B \}$.

We propose to place the accelerometers at $r_0 = [0,0,0]^\top$, $r_1 = [l,0,0]^\top$, $r_2 = [0,l,0]^\top$ and $r_3 = [0,0,l]^\top$ as shown in Fig.~\ref{fig:Tso(3)}. Thus, the difference between accelerometer $i$ and accelerometer $0$ is
\begin{align*}
    a_i - a_0 = \dot{\omega} \times r_i + \Omega^{\times} \Omega^{\times} r_i
\end{align*}
and substituting values of $r_i$,
\begin{align*}
    a_1 - a_0 &= l \left( \begin{bmatrix}
    0 \\
    \theta_z \\
    - \theta_y
    \end{bmatrix} + \begin{bmatrix}
    - {\Omega_y}^2 - {\Omega_z}^2 \\
    \Omega_x \Omega_y \\
    \Omega_x \Omega_z
    \end{bmatrix} \right), \\
    a_2 - a_0 &= l \left( \begin{bmatrix}
    -\theta_z \\
    0 \\
    \theta_x
    \end{bmatrix} + \begin{bmatrix}
    \Omega_x \Omega_y \\
    - {\Omega_x}^2 - {\Omega_z}^2 \\
    \Omega_y \Omega_z
    \end{bmatrix} \right), \\
    a_3 - a_0 &= l \left( \begin{bmatrix}
    \theta_y \\
    - \theta_x \\
    0
    \end{bmatrix} + \begin{bmatrix}
    \Omega_x \Omega_z \\
    \Omega_y \Omega_z \\
    - {\Omega_x}^2 - {\Omega_y}^2
    \end{bmatrix} \right).
\end{align*}
Thus,
\begin{align}
    \theta = \begin{bmatrix}
    \theta_x \\
    \theta_y \\
    \theta_z
    \end{bmatrix} = \frac{1}{2 l}\begin{bmatrix}
    (a_{2,z} - a_{0,z}) - (a_{3,y} - a_{0,y}) \\
    (a_{3,x} - a_{0,x}) - (a_{1,z} - a_{0,z}) \\
    (a_{1,y} - a_{0,y}) - (a_{2,x} - a_{0,x})
    \end{bmatrix}.
\end{align}

\begin{figure}[ht]
    \centering
    \includegraphics[width=0.5\columnwidth]{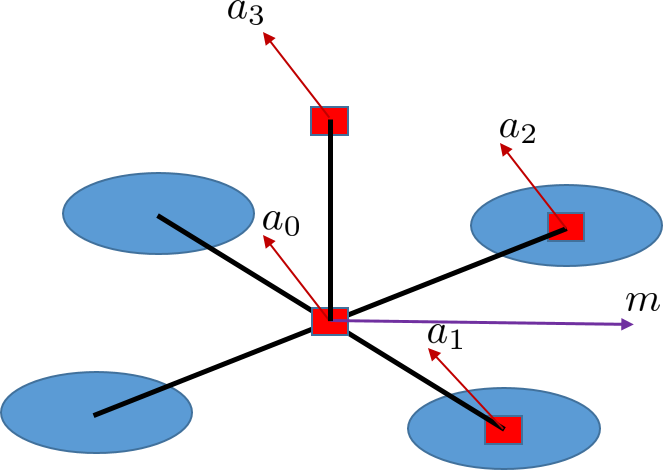}
    \caption{Proposed sensing system. Red boxes are accelerometers taking acceleration measurements $a_i$ at known locations. The inertial sensor in the middle also has a magnetometer taking magnetic field measurement $m$. }
    \label{fig:Tso(3)}
\end{figure}

\subsection{Measurement of Orientation}
Assuming that the magnetic field direction is known and remains constant, the orientation of an object with respect to an inertial frame can be measured by a magnetometer rigidly mounted on the object. For simplicity, the magnetometer is assumed to be aligned (in orientation) with the body-fixed frame $\{ B \}$.

The physical direction of the magnetic field of the earth can be modelled as a direction on the two sphere $S^2$ embedded in $\mathbb{R}^3$. We denote the magnetic field direction relative to the inertial frame by $\mathring{m}$, and measured magnetic field direction in body-fixed frame by $m$, such that $\mathring{m}, m \in S^2 \subset \mathbb{R}^3$. Then
\begin{equation}
    m = R^\top \mathring{m},
\end{equation}
where $R$ is the attitude (orientation) of the object with respect to the inertial frame.
In the presence of multiple magnetometers, assuming that the orientation of the sensors are aligned and are connected rigidly, the measurements can be combined by taking the average.

If we consider an application in Earth's gravity field for an object that is not too dynamic (i.e. most of the time, the object is not accelerating with respect to the Earth fixed frame), then the acceleration of $\{ B \}$ measured at accelerometer $0$ is approximately the gravitational acceleration direction measured in the body-fixed frame.
Thus, the outputs are
\begin{subequations}
\begin{align}
    y_1 &= m = R^\top \mathring{m}, \\
    y_2 &= a_0 \cong g R^\top e_3,
\end{align}
\end{subequations}
where $y_1, y_2 \in \mathbb{R}^3$. 

Note that other application specific sensor that provides vector measurement may also be used instead of the measurement from accelerometer $0$. For example, in satellite application, sun sensor~\cite{Magnis14} can be used. 

\section{SYMMETRY OF ATTITUDE SYSTEM}
Most physical systems have physical models with symmetries that encode the equivariance of the laws of motion. This expresses the fact that the behaviour of the system at one point is the same as the behaviour at another point in the state space, when viewed through a symmetric transformation of space~\cite{Mahony13}. Then, a global observer design can be achieved by analysing the behaviour at one point in space for a system with symmetry. 

\subsection{The Symmetry group}
We introduce a new notation for the inverse mapping from a skew-symmetric matrix to the corresponding vector such that
\begin{equation}
    (\Omega^{\times})^\vee = \Omega
\end{equation}
and that $\so^\vee(3) \equiv \mathbb{R}^3$.

Consider the Lie-group $\SE(3)$.
This group is usually considered as a model for rigid-body transformations of the Euclidean space.
In the following development we show that this group can also be used to model a symmetry of the the second order velocity kinematics for attitude.
Write the elements of $\SE(3) = \SO(3) \ltimes \so^\vee(3)$ in the form
\begin{align*}
    \SE(3) = \{ (Q, q) \;\vline\; Q\in \SO(3), q \in \R^3 \}.
\end{align*}
The semi-direct group product is given by
\begin{align*}
    (A,a) \cdot (B,b) = (AB, a + A b ),
\end{align*}
with group identity $(I_3, 0)$ and group inverse
\begin{align*}
    (A,a)^{-1} = (A^\top, - A^\top a).
\end{align*}

\begin{lemma} \label{lem:phi}
The action $\phi: \SE(3) \times \calM \to \calM$ defined by
\begin{align}
\label{eq:phi_func}
    \phi((Q,q),(R,\Omega^{\times})) = (RQ, (Q^\top(\Omega + q))^{\times})
\end{align}
is a transitive right group action of $\SE(3)$ on $\calM$.
\end{lemma}

\begin{proof}
Let $(A,a),(B,b) \in \SE(3)$.
Then
\begin{align*}
    \phi((A,a), &\phi((B,b),(R, \Omega^{\times}))) \\
    &= \phi((A,a), (RB, (B^\top(\Omega + b))^{\times}) \\
    &= (RBA, (A^\top(B^\top(\Omega + b) + a))^{\times}) \\
    &= (R(BA), (A^\top B^\top(\Omega + b + B a))^{\times}) \\
    &= (R(BA), ((BA)^\top(\Omega + (b + B a)))^{\times}) \\
    &= \phi((BA, b+B a), (R, \Omega^{\times})) \\
    &= \phi((B,b) \cdot (A,a), (R, \Omega^{\times})).
\end{align*}
This demonstrates that the (right handed) group action property holds.
It is straightforward to verify that $\phi((\Id_3,0), (R,\Omega)) = (R,\Omega)$.

To see that $\phi$ is transitive, let $(R,\Omega^{\times})$ and $(R', {\Omega'}^{\times})$ be any elements of $\calM$. Then we can find the group element $(R^\top R', R^\top R'(\Omega') - \Omega)$ such that
\begin{align*}
    \phi((R^\top R', &R^\top R'(\Omega') - \Omega), (R, \Omega^{\times})) \\
    &= (RR^\top R', ((R^\top R')^\top (R^\top R'(\Omega') - \Omega + \Omega))^{\times}) \\
    &= (R', ((R^\top R')^\top (R^\top R') \Omega')^{\times}) \\
    &= (R', \Omega'^{\times}).
\end{align*}
\end{proof}

The key observation from this derivation is that the same Lie-group
$\SE(3)$ can be used as a symmetry group either for rigid-body kinematics, or as a symmetry group for second order attitude kinematics.
Note however that the state-space of the observer problem for attitude kinematics $T \SO(3)$ is not $\SE(3)$.

\subsection{Equivariance of the kinematics}
\begin{lemma} \label{lem:psi}
The action $\psi : \SE(3) \times \vecV \to \vecV$ defined by
\begin{align} \label{eq:psi}
    \psi((Q,q),(\pi^{\times}, \theta^{\times})) = ((Q^\top(\pi - q))^{\times}, (Q^\top \theta)^{\times})
\end{align}
is a right action of $\SE(3)$ on $\vecV$.
\end{lemma}

\begin{proof}
Let $(A,a),(B,b) \in \SE(3)$ and let $(\pi^{\times},\theta^{\times}) \in \vecV$. Then we have
\begin{align*}
    \psi((A,a),&\psi((B,b),(\pi^{\times},\theta^{\times}))) \\
    &= \psi((A,a), ((B^\top(\pi - b))^{\times}, (B^\top \theta)^{\times})) \\
    &= ((A^\top (B^\top(\pi - b) - a))^{\times}, (A^\top B^\top \theta)^{\times}) \\
    &= ((A^\top ( B^\top(\pi - b - B a)))^{\times}, ((BA)^\top \theta)^{\times}) \\
    &= ((A^\top B^\top(\pi - b - B a))^{\times}, ((BA)^\top \theta)^{\times}) \\
    &= \psi((BA, b + B a), (\pi^\times, \theta^\times)) \\
    &= \psi((B,b)\cdot (A,a), (\pi^\times, \theta^\times)).
\end{align*}
This shows that the (right handed) group action property holds. It is straightforward to verify that $\psi((\Id_3,0),(\pi^{\times}, \theta^{\times})) = (\pi^{\times}, \theta^{\times})$. Thus, $\psi$ is indeed a right action.
\end{proof}

Next, we wish to show the system dynamics are equivariant. 

\begin{lemma} \label{lem:equivariant_input}
The system defined in \eqref{eq:system} is equivariant with respect to the group $\SE(3)$ and group actions $\phi$ and $\psi$ as given in Lemmas \ref{lem:phi} and \ref{lem:psi} respectively; i.e., one has
\begin{align*}
    \td \phi_X [f(x, v)] = f(\phi_X(x), \psi_X(v)),
\end{align*}
for any $X \in \SE(3)$, $x \in \calM$, and $v \in \vecV$
\end{lemma}
\begin{proof}
Let $(Q,q) \in \SE(3)$, $(R,\Omega^{\times}) \in \calM$, and $(\pi^{\times}, \theta^{\times}) \in \vecV$ be arbitrary.
Then we have 
\begin{align*}
    \td \phi_{(Q,q)}  f&((R,\Omega^{\times}), (\pi^{\times},\theta^{\times})) \\
    &= \tD_{(R,\Omega)} \phi_{(Q,q)} (R,\Omega^{\times}) [(R(\pi + \Omega)^{\times}, \theta^{\times})] \\
    &= (R(\pi + \Omega)^{\times} Q , (Q^\top \theta)^{\times} ) \\
    &= (RQ Q^\top(\pi + \Omega)^{\times} Q, (Q^\top \theta)^{\times}) \\
    &= (RQ (Q^\top(\pi + \Omega))^{\times}, (Q^\top \theta)^{\times}) \\
    &= (RQ ((Q^\top(\Omega + q))^{\times} + (Q^\top(\pi - q))^{\times}), (Q^\top \theta)^{\times}) \\
    &= f( (RQ, (Q^\top(\Omega + q))^{\times}), ((Q^\top(\pi - q))^{\times}, (Q^\top \theta)^{\times})) \\
    &= f( \phi_{(Q,q)} (R,\Omega^{\times}), \psi_{(Q,q)} (\pi^{\times}, \theta^{\times}) ),
\end{align*}
which proves the equivariance condition.
\end{proof}

Finally, we demonstrate that the group action on the output is also equivariant.
\begin{lemma} \label{lem:rho}
The action $\rho^i: \SE(3) \times \N_i \to \N_i$ defined by
\begin{equation*}
    \rho^i((Q,q), y_i) = Q^\top y_i
\end{equation*}
is a transitive right group action of $\SE(3)$ on $\N_i$, where $\N_i \equiv \mathbb{R}^3$.
\end{lemma}
\begin{proof}
Let $(A, a), (B, b) \in \SE(3)$ and let $y_i \in \N$. Then we have
\begin{align}
    \rho^i((A, a), \rho^i((B, b), y_i)) &= \rho^i((A, a), B^\top y_i)), \nonumber \\
    &= A^\top B^\top y_i, \nonumber \\
    &= (BA)^\top y_i, \nonumber \\
    &= \rho^i((B, b) \cdot (A, a), y_i). \nonumber
\end{align}
This shows that the (right handed) group action property holds. It is straightforward to verify that $\rho^i((I_3, 0), y_i) = y_i$. Thus, $\rho$ is indeed a right action.
\end{proof}

\subsection{System lift onto the group}
In order to work with our observer on $\SE(3)$, we need to lift the kinematics of the system. A lifted equivariant system is defined as the system on the $\SE(3)$-torsor
\begin{subequations}
\begin{align}
    \dot{X} &:= F_{x_0}(X, v), \\
    y_i &:= \rho^i(X, \mathring{y}_i)  =: H_i(x),
\end{align}
\end{subequations}
for $v \in \vecV$ and initial condition $X(0) \in \SE(3)$ such that $\phi_{x_0}(X(0)) = x(0)$ projects to the initial condition of (\ref{eq:kinematics}).

Specifically, we require a system lift, which is a family of functions $F_{x_0}: \SE(3) \times \vecV \to T\SE(3)$ such that $\td\phi_{x_0} [F_{x_0}(X, v)] = f(\phi_{x_0}(X), v)$ for all $x_0 \in \calM$, $v \in \vecV$, and $X \in \SE(3)$. Lemma \ref{lem:system_lift} provides such a family.

\begin{lemma} \label{lem:system_lift}
The family of functions $F_{x_0}: \SE(3) \times \vecV \to T\SE(3)$ defined by
\begin{align}
    F&_{(R_0, (\Omega_0)^{\times})}((Q,q), (\pi^{\times}, \theta^{\times})) \nonumber \\
    &= (Q(Q^\top (\Omega_0 + q) + \pi)^{\times}, \ Q (\pi^{\times} Q^\top(\Omega_0 + q) + \theta)) \nonumber \\
    &= ( (\Omega_{0} + q)^{\times} Q + (Q \pi)^{\times} Q, (Q \pi)^{\times} (\Omega_{0} + q) + Q \theta )
    \label{eq:system_lift}
\end{align}
defines a system lift for \eqref{eq:system} onto the group $\SE(3)$.
\end{lemma}
\begin{proof}
In order to show that $\{ F_{x_0} \}_{x_0 \in \calM}$ is indeed a system lift, we need to show that $\td\phi_{x} F_{x_0}(X, v) = f(\phi_{x_0}(X), v)$ for all $x_0 \in \calM$, $v \in \vecV$, and $X \in \SE(3)$. Let $(R_0, (\Omega_0)^{\times}) \in \calM$, $(\pi^{\times}, \theta^{\times}) \in \vecV$, and $(Q,q) \in \SE(3)$ be arbitrary.

Then we have
\begin{align*}
    \td\phi&_{(R_0, (\Omega_0)^{\times})} F_{(R_0, (\Omega_0)^{\times})}((Q,q), (\pi^{\times}, \theta^{\times})) \\
    &= \td\phi_{(R_0, (\Omega_0)^{\times})} [Q(Q^\top (\Omega_0 + q) + \pi)^{\times}, \\
    &\hspace{0.5cm} Q (\pi^{\times} Q^\top(\Omega_0 + q) + \theta)] \\
    &= \frac{\td}{\td t} \phi_{(R_0, (\Omega_0)^{\times})} (Q \exp(t(Q^\top (\Omega_0 + q) + \pi)^{\times} ), \\
    &\hspace{0.5cm} Q (t (\pi^{\times} Q^\top(\Omega_0 + q) + \theta)) + q)|_{t=0} \\
    &= \frac{\td}{\td t} (R_0 Q \exp(t(Q^\top (\Omega_0 + q) + \pi)^{\times} ), \\
    &\hspace{0.5cm} [(Q \exp(t(Q^\top (\Omega_0 + q) + \pi)^{\times} ))^\top (\Omega_0 +  \\
    & \hspace{0.5cm}(\Omega_0 + Q (t (\pi^{\times} Q^\top(\Omega_0 + q) + \theta)) + q ) ]^{\times} )|_{t=0} \\
    &= (R_0 Q (Q^\top (\Omega_0 + q) + \pi)^{\times} , \\
    &\hspace{0.5cm} [ -(Q^\top (\Omega_0 + q) + \pi)^{\times} Q^\top (\Omega_0 + q) +\\
    & \hspace{0.5cm}  Q^\top Q (\pi^{\times} Q^\top(\Omega_0 + q) + \theta))  ]^{\times} ) \\
    &= (R_0Q (Q^\top (\Omega_0 + q) + \pi)^{\times} , \\
    & \hspace{0.5cm}[ -(Q^\top (\Omega_0 + q))^{\times} Q^\top (\Omega_0 + q) - \pi^{\times} Q^\top (\Omega_0 + q) + \\
    & \hspace{0.5cm} \pi^{\times} Q^\top(\Omega_0 + q) + \theta  ]^{\times} ) \\
    &= (R_0 Q (Q^\top (\Omega_0 + q) + \pi)^{\times} , \ \theta^{\times}) \\
    &= f( (R_0 Q, (Q^\top(\Omega_0+q))^{\times}), (\pi^{\times}, \theta^{\times})) \\
    &= f( \phi_{(R_0, (\Omega_0)^{\times})}((Q,q)), (\pi^{\times}, \theta^{\times}) ),
\end{align*}
which shows that our condition holds, and therefore $\{ F_{x_0} \}_{x_0 \in \calM}$ is a system lift as required.
\end{proof}

Let an element of the lifted state space be represented as $X = (Q, q) \in \SE(3)$, with unknown initial condition $X(0) = (Q_0, q_0)$.
Lemma~\ref{lem:system_lift} shows that the lifted kinematics on $\SE(3)$ can be written as
\begin{subequations}
\label{eq:lifted_system_kinematics}
\begin{align}
    \dot{Q} &= (\Omega_{0} + q)^{\times} Q + (Q \pi)^{\times} Q, \\
    \dot{q} &= (Q \pi)^{\times} (\Omega_{0} + q) + Q \theta,
\end{align}
\end{subequations}
where $\pi \equiv \mathbf{0}$ in the real system, since there is no input angular velocity (unlike first order state kinematics).

\section{OBSERVER DESIGN}
From equation (\ref{eq:full_kinematics}), assuming no external angular velocity input ($\pi \equiv \mathbf{0}$), the behaviour of the second order attitude kinematics satisfies 
\begin{subequations}
\begin{align}
    \dot{R} &= R \Omega^{\times}, \\
    \dot{\Omega} &= \theta. 
\end{align}
\end{subequations}
The (output) measurements available are 
\begin{align*}
    a &= \frac{a_0}{|a_0|} = R^\top \mathring{a} \cong R^\top e_3, \\
    \hat{a} &= \hat{R}^\top \mathring{a} \cong \hat{R}^\top e_3, \\
    m &= R^\top \mathring{m}, \\
    \hat{m} &= \hat{R}^\top \mathring{m},
\end{align*}
where a, $\hat{a}, m, \hat{m} \in S^2$, and $|a| = |\hat{a}| = |m| = |\hat{m}| = 1$.

\subsection{Proposed Observer}
Our proposed observer works on the symmetry group of the attitude state space i.e.\ $\SE(3)$, and not the manifold $\calM$.
We use $\hat{X}(t; \hat{X}(0)) \in \SE(3)$ to denote the estimate for the lifted system state $X(t; X(0))$ for unknown $X(0)$. The fundamental structure for the observer that we consider is that of a \emph{pre-observer} (a copy of (\ref{eq:lifted_system_kinematics})) with innovation.
The innovation takes outputs $\{ y_i \}$ and the observer state $\hat{X}$, and generates a correction term for the observer dynamics with the goal that $\hat{x} = \phi(\hat{X}, x_0)$ converges to $x(t, x_0)$.

\begin{theorem}
\label{th:observer_lifted}
Let $(R_0, \Omega_0)$ be a constant reference state in $\calM$. 
Let $\hat{X} = (\hat{Q}, \hat{q}) \in \SE(3)$, with arbitrary initial condition $\hat{X}(0) = (\hat{Q}_0, \hat{q}_0)$.
Consider the  observer kinematics
\begin{subequations}
\label{eq:lifted_observer_kinematics}
\begin{align}
    \dot{\hat{Q}} &= (\Omega_{0} + \hat{q})^{\times} \hat{Q} + (\hat{Q} \hat{\pi})^{\times} \hat{Q}, \\
    \dot{\hat{q}} &= (\hat{Q} \hat{\pi})^{\times} (\Omega_{0} + \hat{q}) + \hat{Q} \hat{\theta},
\end{align}
\end{subequations}
where
\begin{align*}
    \hat{\pi} &= k_1 ( m \times \hat{Q}^\top{R_0}^\top \mathring{m} + a \times \hat{Q}^\top{R_0}^\top \mathring{a}), \\
    \hat{\theta} &= \theta + k_2 ( m \times \hat{Q}^\top{R_0}^\top \mathring{m} + a \times \hat{Q}^\top{R_0}^\top \mathring{a}).
\end{align*}
Denote $(\hat{R}, \hat{\Omega}) = \phi((\hat{Q},\hat{q}), (R_0, \Omega_0))$, then $(\hat{R}, \hat{\Omega})$ converges uniformly locally exponentially, and almost globally asymptotically to $(R, \Omega)$. 
\end{theorem}

\begin{proof}
From Lemma 2 of \cite{Mahony13}, the solution $X(t; X_0)$ projects back to the state $x(t; x_0)$ via the group action
\begin{align*}
    \phi_{x_0}(X (t; X(0)) = x(t; x_0) 
\end{align*}
such that
\begin{align}
\label{eq:projection_to_state_space}
    (R, \Omega^{\times}) &= \phi((Q,q), (R_0, (\Omega_0)^{\times})) \nonumber \\
    &= (R_0 Q, (Q^\top (\Omega_0 + q))^{\times}),
\end{align}
and,
\begin{align}
    (\hat{R}, \hat{\Omega}^{\times}) &= \phi((\hat{Q},\hat{q}), (R_0, (\Omega_0)^{\times})) \nonumber \\
    &= (R_0 \hat{Q}, (\hat{Q}^\top (\Omega_0 + \hat{q}))^{\times}).
    \label{eq:hatX_to_hatx}
\end{align}

This can be intuitively understood as shown in Fig.~\ref{fig:original_lifted_state}.
\begin{figure}[ht]
\centering
\includegraphics[width=.6\columnwidth]{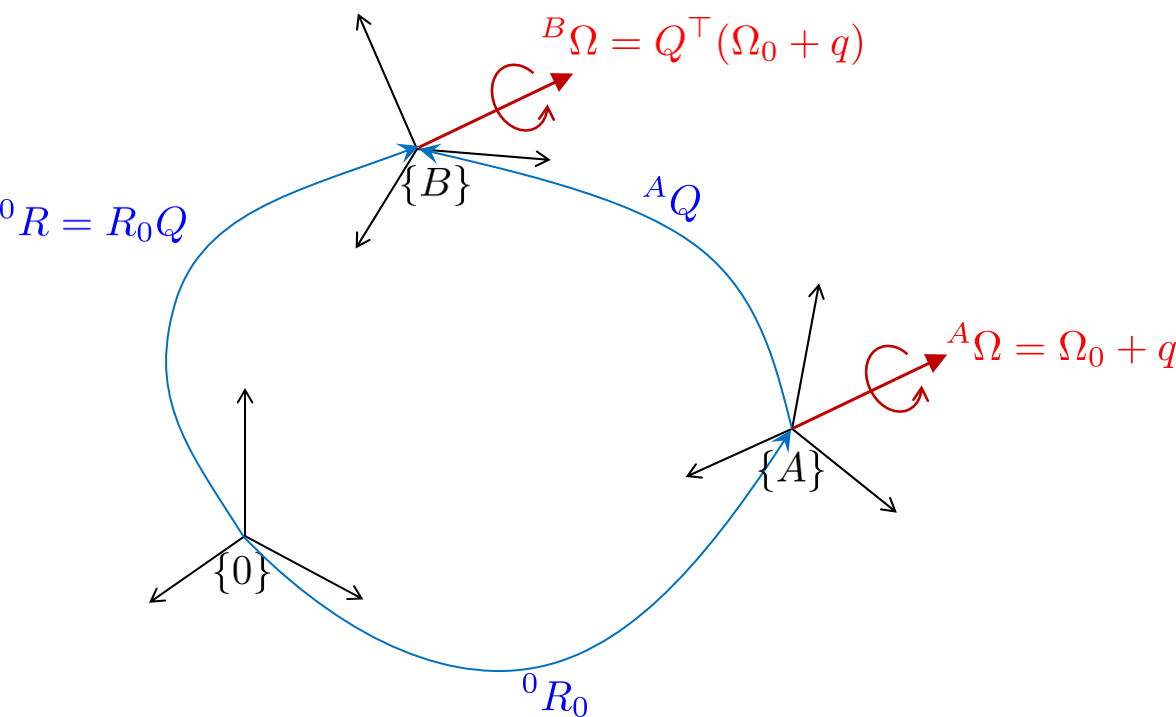}
\caption{\label{fig:original_lifted_state}Relationship between original state $(^{0}R, ^{B}\Omega)$, reference state $(^{0}R_0, ^{A}\Omega_0)$ and lifted state $(^{A}Q,^{A}q)$. Note that the superscript before the variables denotes the frame of reference, and is left out when the frame of reference is clear. For the estimated variables, the ground truth variables are simply replaced by the estimate (\emph{e.g.} $R$ is replaced by $\hat{R}$ and so on). }
\end{figure}

The stability of observer (\ref{eq:lifted_observer_kinematics}) is analysed as follows. The Lyapunov function is defined as
\begin{equation}
\label{eq:Lyapunov}
    \mathcal{L} = (1 - \hat{a}^\top a) + (1 - \hat{m}^\top m) + \frac{1}{2 k_2}|\Omega - \hat{\Omega}|^2 .
\end{equation}

Let $\alpha = ((m \times \hat{m}) + (a \times \hat{a}))$, $\tilde{\Omega} = \Omega - \hat{\Omega}$.
Computing the derivative of the (output) measurements, we get
\begin{align*}
    \dot{a} &= - \Omega^{\times} a \\
    \dot{\hat{a}} &= - (\hat{\Omega} + k_1 \alpha)^{\times} \hat{a} \\
    \dot{m} &= - \Omega^{\times} m \\
    \dot{\hat{m}} &= - (\hat{\Omega} + k_1 \alpha)^{\times} \hat{m}
\end{align*}

From (\ref{eq:hatX_to_hatx}), computing the time derivatives, the proposed observer in $\SE(3)$ when projected back into the original state space in $\calM$ has kinematics
\begin{align*}
    \dot{\hat{R}} &= R_0 \dot{\hat{Q}} \\
    &= R_0 \left( (\Omega_0 + \hat{q})^{\times} \hat{Q} + (\hat{Q} \hat{\pi})^{\times} \hat{Q} \right) \\
    &= R_0 \left( \hat{Q} (\hat{Q}^\top(\Omega_0 + \hat{q}))^{\times} + \hat{Q}(\hat{\pi})^{\times} \right) \\
    &= \hat{R} (\hat{\Omega} + \hat{\pi})^{\times}
\end{align*}
and,
\begin{align*}
    \dot{\hat{\Omega}} &= \dot{\hat{Q}}^\top (\Omega_0 + \hat{q}) + \hat{Q}^\top \dot{\hat{q}} \\
    &= - (\hat{Q}^\top (\hat{Q} \hat{\pi})^{\times} + \hat{Q}^\top (\Omega_0 + \hat{q})^{\times})(\Omega_0 + \hat{q}) \\
    &\hspace{0.5cm}+ \hat{Q}^\top ((\hat{Q} \hat{\pi})^{\times} (\Omega_0 + \hat{q}) + \hat{Q} \hat{\theta}) \\
    &= - \hat{Q}^\top (\hat{Q} \hat{\pi})^{\times} (\Omega_0 + \hat{q}) + \hat{Q}^\top (\hat{Q} \hat{\pi})^{\times} (\Omega_0 + \hat{q}) + \hat{\theta} \\
    &= \hat{\theta}
\end{align*}

The derivative of the Lyapunov function is
\begin{align*}
    \dot{\mathcal{L}} &= -  (\dot{\hat{a}}^\top a + \hat{a}^\top \dot{a} + \dot{\hat{m}}^\top m + \hat{m}^\top \dot{m}) + \frac{1}{k_2} (\Omega - \hat{\Omega})^\top(\dot{\Omega} - \dot{\hat{\Omega}}) \\
    &= -  (\hat{a}^\top (\hat{\Omega} + k_1 \alpha)^{\times} a - \hat{a}^\top \Omega^{\times} a + \hat{m}^\top (\hat{\Omega} + k_1 \alpha)^{\times} m - \hat{m}^\top \Omega^{\times} m) \\
    & \hspace{0.5cm}+ \frac{1}{k_2}(\Omega - \hat{\Omega})^\top(- k_2 \alpha) \\
    &= -  (\hat{a}^\top (k_1 \alpha + \tilde{\Omega})^{\times} a + \hat{m}^\top (k_1 \alpha + \tilde{\Omega})^{\times} m) + \tilde{\Omega}^\top \alpha \\
    &= -  (-\hat{a}^\top a^{\times} (k_1 \alpha + \tilde{\Omega}) - \hat{m}^\top m^{\times} (k_1 \alpha + \tilde{\Omega})) + \tilde{\Omega}^\top \alpha \\
    &= -  ((a \times \hat{a})^\top (k_1 \alpha + \tilde{\Omega}) + (m \times \hat{m})^\top (k_1 \alpha + \tilde{\Omega})) + \tilde{\Omega}^\top \alpha \\
    &= - ((a \times \hat{a} + m \times \hat{m})^\top (k_1 \alpha + \tilde{\Omega})) + \tilde{\Omega}^\top \alpha \\
    &= -  k_1 (\alpha^\top \alpha) - \alpha^\top \tilde{\Omega} + \tilde{\Omega}^\top \alpha \\
    &= -  k_1 ||\alpha||^2 \\
    &= -  k_1 ||(m \times \hat{m}) + (a \times \hat{a})||^2
\end{align*}

We define $\Tilde{R} \triangleq \hat{R} R^\top$, and use identities
\begin{align*}
    a^\top b = \frac{1}{2} tr({a^{\times}}^\top b^{\times}), \\
    (a \times b)^{\times} = b a ^\top - a b^\top.
\end{align*}
Then,
\begin{align}
\label{eq:observer1_dLyapunov}
    \dot{\mathcal{L}} &= -  k_1 ||R^\top (\mathring{m} \times (\Tilde{R}^\top \mathring{m}) + \mathring{a} \times (\Tilde{R}^\top \mathring{a}))||^2  \nonumber \\
    &= -  k_1 ||\mathring{m} \times (\Tilde{R}^\top \mathring{m}) + \mathring{a} \times (\Tilde{R}^\top \mathring{a})||^2  \nonumber \\
    &= -\frac{1}{2}  k_1 tr((\Tilde{R}^\top \mathring{m} \mathring{m}^\top - \mathring{m} \mathring{m}^\top \Tilde{R} + \Tilde{R}^\top \mathring{a} \mathring{a}^\top - \mathring{a} \mathring{a}^\top \Tilde{R})^{\times \top} \nonumber \\
    &\hspace{0.5cm}(\Tilde{R}^\top \mathring{m} \mathring{m}^\top - \mathring{m} \mathring{m}^\top \Tilde{R} + \Tilde{R}^\top \mathring{a} \mathring{a}^\top - \mathring{a} \mathring{a}^\top \Tilde{R})^{\times})  \nonumber \\
    &= -\frac{1}{2}  k_1 tr((\Tilde{R}^\top M - M \Tilde{R})^{\times \top} (\Tilde{R}^\top M - M \Tilde{R})^{\times})  \nonumber \\
    &= - k_1 ((\Tilde{R}^\top M - M \Tilde{R})^\top (\Tilde{R}^\top M - M \Tilde{R}))  \nonumber \\
    &= - k_1 ||\Tilde{R}^\top M - M \Tilde{R}||^2
\end{align}
with $M = \mathring{m} \mathring{m}^\top + \mathring{a} \mathring{a}^\top$ (a positive definite matrix), is negative semi-definite. This proves that the estimation error terms are bounded. Analogously to Theorem 4.2 in~\cite{Mahony08}, the computation of the second time derivative of the Lyapunov function
\begin{align*}
\ddot{\mathcal{L}} = -k_1 (\Tilde{R}^\top M - M \Tilde{R})^\top (\dot{\Tilde{R}}^\top M - M \dot{\Tilde{R}}),
\end{align*}
where we assume $\Omega$ is bounded (matter cannot exceed the speed of light), such that $\tilde{\Omega}$ is also bounded. It is clear that $(\tilde{R}\mathring{m} \times \mathring{m} + \tilde{R}\mathring{a} \times \mathring{a})$ is bounded. 
With direct application of Barbalat's lemma, proves that $\dot{\mathcal{L}}$ asymptotically converges to zero, and hence
\begin{align}
    \Tilde{R}^\top M &= M \Tilde{R}. \label{eq:RMMR}
\end{align}
This in turn implies that the equilibrium $(\tilde{R}, \tilde{\Omega})=(I,0)$ is asymptotically stable. The unstable equilibrium set is $(\{(\tilde{R}, \tilde{\Omega})= (\{\tilde{R} | U_0 D_i {U_0}^\top \}, 0)$ with $D_i$ and $U_0$ defined in Theorem 5.1 of~\cite{Mahony08}. 

The proof of the local exponential stability for the stable equilibrium $(\tilde{R}, \tilde{\Omega}) = (I,0)$, is the same as the proof of item 2 of Theorem 5.1 \cite{Mahony08} when using $(\tilde{R}, R\tilde{\Omega})$ as state variables. 

\end{proof}
Note that, the observer error on $\calM$ has local uniform exponential convergence, and almost global asymptotic convergence property. Similarly, regardless of the chosen reference $(R_0, \Omega_0)$, our proposed observer error on $\SE(3)$ has local uniform exponential convergence, and almost global asymptotic convergence property.

\section{SIMULATION RESULTS}
We evaluate the performance of our proposed observers using simulation. We consider an object with four accelerometers and one magnetometer placed at known locations as shown in Fig.~\ref{fig:Tso(3)}, with the distance from the centre $l = 1$. The sensors have simulated additive Gaussian noise with standard deviation of $0.3 m/s^2$ and $0.3$ on all axes for accelerometer measurement and the magnetic field direction respectively. The chosen simulation time step $dt = 0.001 s$, which corresponds to sensors' sampling frequency of $1 kHz$.

The object is rotating at non-constant angular velocity, following an arbitrarily chosen function of $\Omega = (\sin(0.1 t), \cos(0.1 t), 1)^\top$. The ground truth attitude is then generated by integration where $R(k+1) = R(k)\exp(\Omega^{\times} dt)$.

The observer gains used are $k_1 = 3$, $k_2 = 1$. The observer on $\SE(3)$ have $(R_0, \Omega_0)$ set to random rotational matrix and random $3 \times 1$ vector respectively, while $\hat{X}(0) = (\hat{Q}_0,\hat{q}_0)$ are also set to random rotational matrix and random $3 \times 1$ vector respectively. 

\begin{figure}[ht]
\centering
\includegraphics[width=0.49\columnwidth]{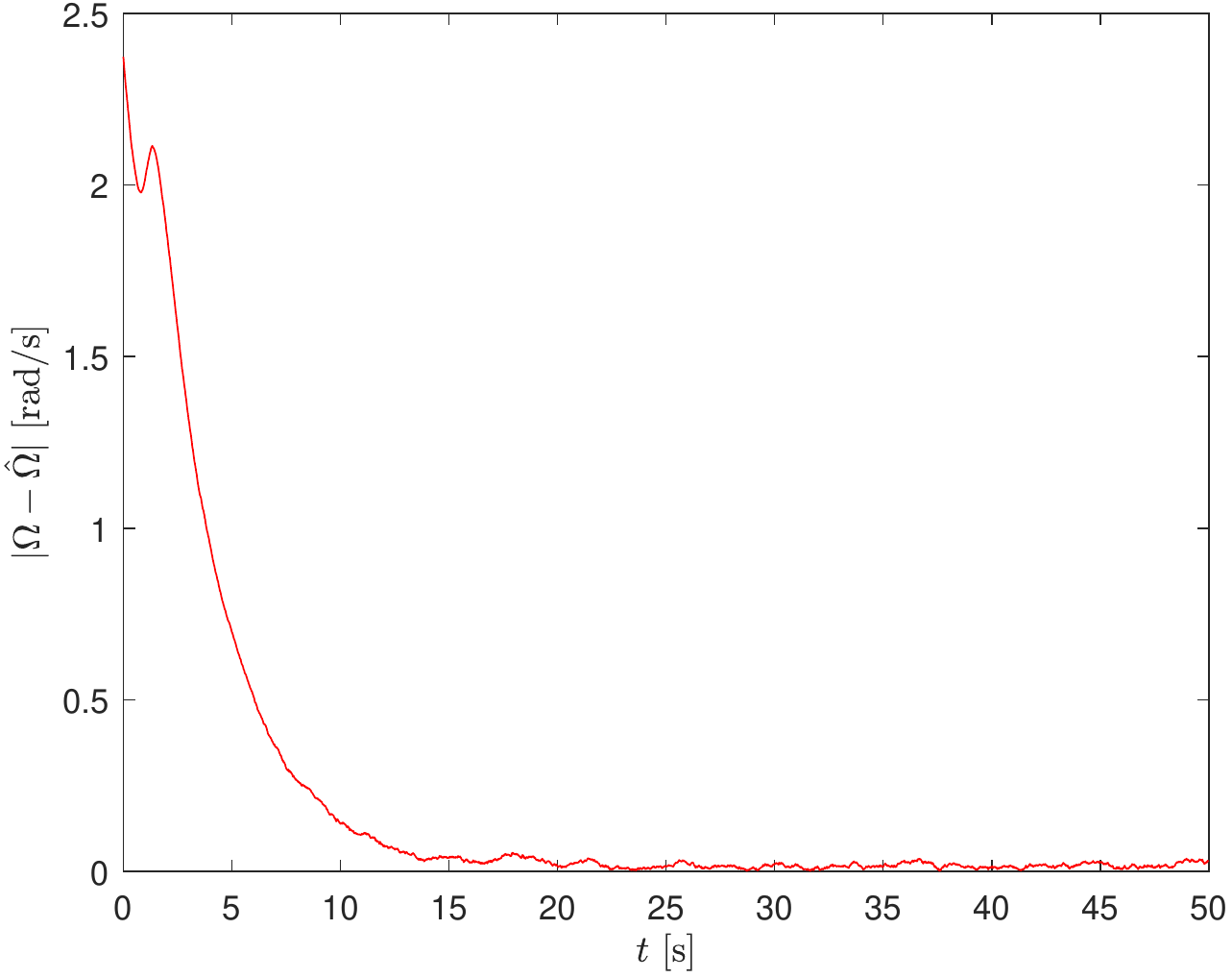}
\includegraphics[width=0.49\columnwidth]{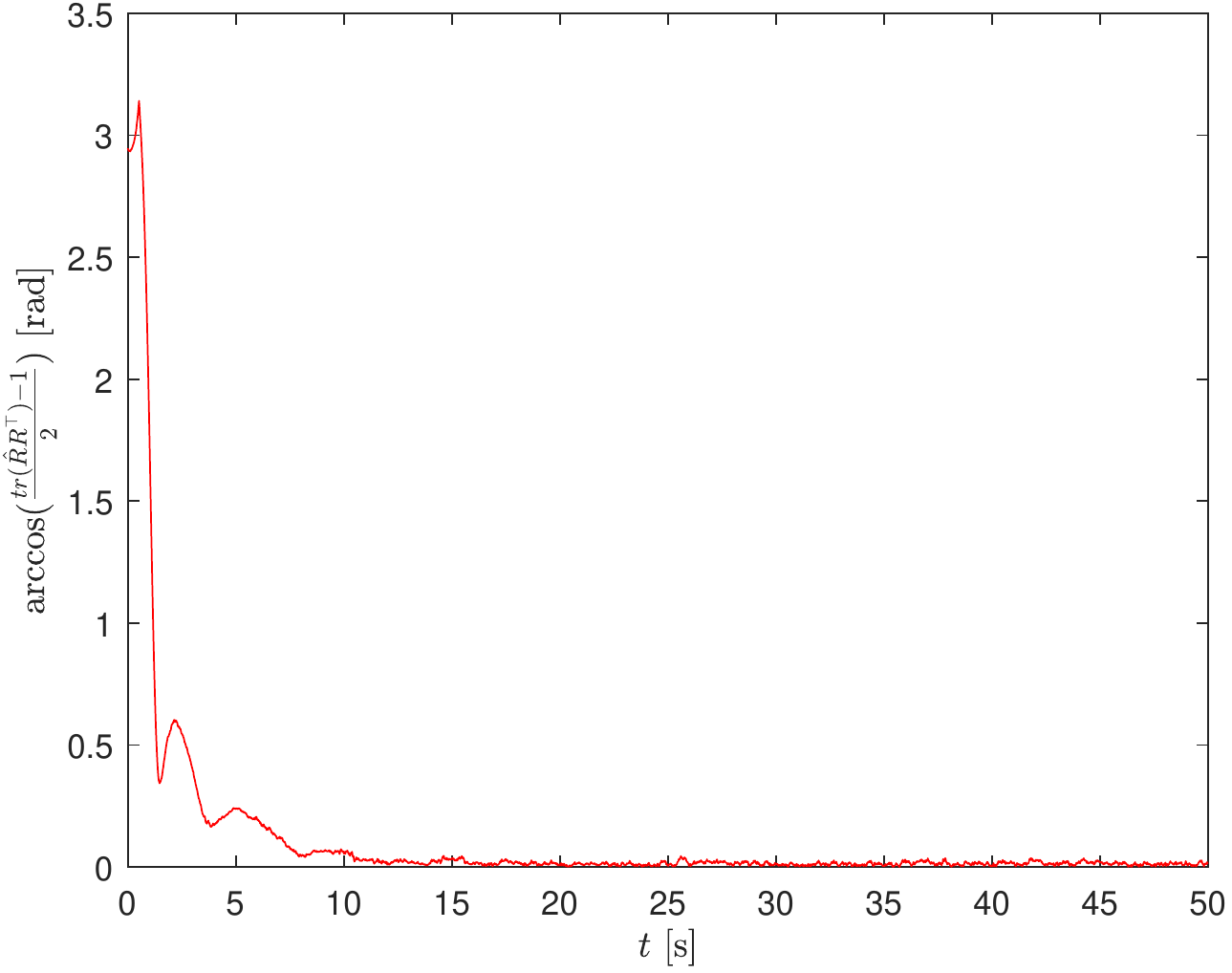}
\caption{\label{fig:error}Time evolution of error showing convergence of estimates $(\hat{R}, \hat{\Omega}^{\times}) = \phi((\hat{Q},\hat{q}), (R_0, (\Omega_0)^{\times}))$ to ground truth $(R, \Omega^\times)$. From left to right: (a) angular velocity error; (b) attitude error}
\end{figure}

\begin{figure}[ht]
\centering
\includegraphics[width=0.6\columnwidth]{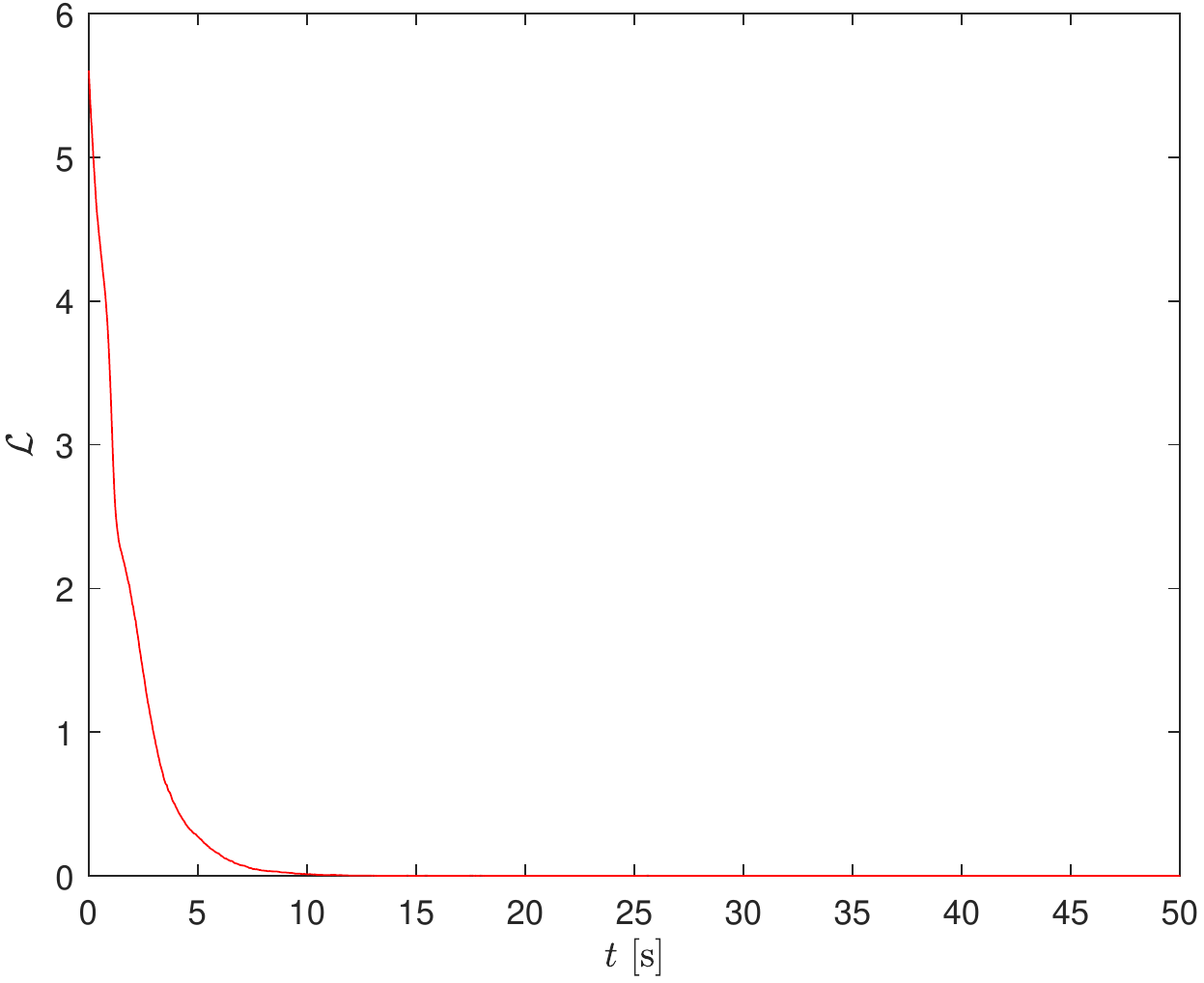}
\caption{\label{fig:Lyapunov}Time evolution of Lyapunov function $\mathcal{L}$ showing local exponential convergence to zero. }
\end{figure}

\section{CONCLUSIONS}
An equivariant observer for the second order attitude kinematics was presented. The observer does not rely on angular velocity input commonly used by existing methods. Instead, a novel sensor modality is proposed, which uses four or more accelerometer to obtain the measurement of angular acceleration.
The observer works on the symmetry group of the second order attitude state space, which has a similar structure as the well known special Euclidean group $\SE(3)$.
The observer is shown to exhibit local uniform exponential convergence, and almost global asymptotic convergence to zero estimation error in the presence of sensor noise, regardless of the chosen reference state $(R_0, \Omega_0)$.

Possible extensions of this work include the study of the effect of sensor bias, similar to the work presented in~\cite{Mahony08}. A nonlinear second order pose observer can also be explored, where the position and linear velocity are also simultaneously estimated. This may be achieved by incorporating GPS measurements or bearing measurements from camera sensors.


\bibliographystyle{plain}

\end{document}